\documentclass[10pt,conference]{IEEEtran}
\usepackage{cite}
\usepackage{amsmath,amssymb,amsfonts}
\usepackage{algorithmic}
\usepackage{graphicx}
\usepackage{textcomp}
\usepackage{xcolor}
\usepackage{subfigure}
\usepackage{multirow}
\usepackage{fontenc}

\usepackage{amsthm}
\newtheorem{theorem}{Theorem}
\newtheorem{lemma}{Lemma}
\newtheorem{corollary}{Corollary}

\def\BibTeX{{\rm B\kern-.05em{\sc i\kern-.025em b}\kern-.08em
    T\kern-.1667em\lower.7ex\hbox{E}\kern-.125emX}}
\begin{document}

\title{On the Credibility of Information Flows in Real-time Wireless Networks\\
}

\author{\IEEEauthorblockN{Daojing Guo}
\IEEEauthorblockA{\textit{Electrical and Computer Engineering Department} \\
\textit{Texas A\&M University}\\
College Station, United States \\
daojing\_guo@tamu.edu}
\and
\IEEEauthorblockN{I-Hong Hou}
\IEEEauthorblockA{\textit{Electrical and Computer Engineering Department} \\
\textit{Texas A\&M University}\\
College Station, United States \\
ihou@tamu.edu}
}

\maketitle
\begin{abstract}\label{section:abstract}
This paper considers a wireless network where multiple flows are delivering status updates about their respective information sources. An end user aims to make accurate real-time estimations about the status of each information source using its received packets. As the accuracy of estimation is most impacted by events in the recent past, we propose to measure the \emph{credibility} of an information flow by the number of timely deliveries in a window of the recent past, and say that a flow suffers from a \emph{loss-of-credibility} (LoC) if this number is insufficient for the end user to make an accurate estimation.

We then study the problem of minimizing the system-wide LoC in wireless networks where each flow has different requirement and link quality. We show that the problem of minimizing the system-wide LoC requires the control of temporal variance of timely deliveries for each flow. This feature makes our problem significantly different from other optimization problems that only involves the average of control variables. Surprisingly, we show that there exists a simple online scheduling algorithm that is near-optimal. Simulation results show that our proposed algorithm is significantly better than other state-of-the-art policies.

\end{abstract}
\section{Introduction}

Many emerging applications, such as industrial Internet of Things (IoT) and virtual reality (VR), require the real-time delivery of information. From an end user's perspective, the performance of such applications are determined by their ability to accurately estimate the real-time status of their respective information sources, such as the temperature of a machine in industrial IoT or the location of a monster in a VR game. However, most existing network performance metrics, ranging from traditional quality-of-service (QoS) metrics such as throughput, delay, and jitter, to emerging ones like timely-throughput and age-of-information, fail to directly capture the accuracy of the users' estimation. Therefore, network algorithms aiming at optimizing these network performance metrics may result in poor performance for these emerging applications.

To address the need for these emerging applications, we introduce the concept of \emph{credibility} of information flows, where an information flow is considered to be credible if its user can make an accurate estimate of the current status using its delivered packets. Our model for credibility is based on two important features of estimation algorithms: First, information in the recent past is much more useful than that in the distant past for making an accurate estimate. Second, most estimation algorithms, even simple ones like linear extrapolation, require multiple data points in the recent past. 

Based on these observations, we propose a model to capture the credibility of information flows in real-time wireless networks. In this model, we consider that each information source, such as sensors generating readings or VR servers generating video frames, generates real-time information periodically. Stale information is dropped in favor of the transmission of new information. The credibility of an information flow only depends on the number of packets that are delivered on time in a window of recent past. If the number of timely deliveries in this window of recent past is below a user-specified threshold, then the estimation becomes inaccurate, and the information flow suffers from a loss-of-credibility (LoC). Our goal is to minimize the system-wide LoC in a wireless network with multiple flows, each with different threshold and channel reliability.

Using Brownian approximation and martingale theory, we show that the problem of minimizing the system-wide LoC is equivalent to an optimization problem that involves two sets of constraints: One set of constraints are related to the average of timely deliveries of each flow, and another set of constraints are related to the variance of timely deliveries. The existence of constraints about the variance of timely deliveries makes this problem significantly different from other network utility maximization (NUM) problems that only involve constraints about the average of variables, and hence cannot be solved by most existing techniques for NUM problems.

We propose a simple online scheduling algorithm for this problem. We analytically prove that the timely deliveries under our scheduling algorithm satisfy both the constraints on the average and those on the variance in the optimization problem. We also analytically prove that our algorithm is near-optimal for the optimization problem in the sense that its performance can be made arbitrarily close to a theoretical bound.

We further evaluate the performance of our algorithm by comparing it against two other state-of-the-art policies, one of them is provably optimal in terms of timely-throughput, and the other achieves an approximation bound in terms of age-of-information. Simulation results show that our policy achieves much smaller LoC than these two policies. This result further highlights that existing network performance metrics may be misleading in capturing the credibility of information flow.

The rest of this paper is as following order:  Section \ref{section:model} introduces our model for credibility in real-time wireless networks. Section \ref{section:formulation} shows that the problem of minimizing LoC is equivalent to an optimization problem. Section \ref{section: policy} introduces our online scheduling algorithm. Section \ref{section:performance} analyzes the performance of our scheduling algorithm and shows that it is near-optimal for the optimization problem. Section \ref{sec:simulation} presents our simulation results. Section \ref{section: related work} reviews some related work. Finally, Section \ref{section: conclusion} concludes this paper. 
\section{System Model}\label{section:model}

We extended the model in \cite{HOU2016ShortTermPerformance}, which focuses on the short-term performance for wireless networks with homogeneous links, to address the credibility of information flows in real-time wireless network where different wireless links can have different channel qualities. 

We consider a real-time wireless network that serves $\mathcal{N}$ clients. Time is slotted, and the duration of one time slot is the amount of time needed by a whole transmissions, including all overheads such as the transmission of poll packet or ACK. Hence, the AP can transmit to at most one client at each time slot, and it has the instantaneous feedback information on whether the transmission is successful. We consider that wireless transmissions are subject to effects of shadowing, multi-path, fading, interference, etc., and different clients experience different channel qualities as they are located at different positions. Hence, we assume that each transmission for client $i$ is successful with probability $p_i$. 

We consider that each client is associated with a real-time information flow, and use flow $i$ to indicate the flow associated with client $i$. Specifically, we assume that each real-time flow generates one packet periodically every $\tau$ slots, that is at time slots 1, $\tau+1$, $2\tau+1$, $\dots$. Each packet has a stringent delay bound of $\tau$ slots, and is removed from the system if it cannot be delivered before its delay bound. In other words, each packet in a real-time flow is only valid for transmission until the next packet arrives. We thereby say that $\tau$ time slots from an \emph{interval}. Packets arrive at the system at the beginning of each interval, and have deadlines at the end of the interval. 

We note that this model for real-time flows applies to many emerging wireless applications. For example, consider multi-user virtual reality (VR) or augmented reality (AR), where an AP streams VR/AR contents to multiple VR/AR headsets. All headsets play VR/AR contents at the same frame rates, and therefore they generate traffic at the same frequency. Further, as the AP should always transmit the newest VR/AR content to a headset, packets that fail to be delivered on time should be removed and replaced by newer packets. Likewise, one can also consider industrial Internet of Things (IoT), where an AP polls measurements from multiple sensors monitoring different locations. Sensors have the same sampling frequency and therefore generate traffic at the same frequency. Also, stale measurements should be dropped when a new measurement is generated. 

An important feature of real-time application such as VR/AR and industrial IoT is that each flow can typically tolerate a small amount of sporadic packet losses, but is very sensitive to a burst of packet losses. For example, in industrial IoT, a controller can use various estimation techniques to estimate the value of a lost sensor reading. However, the accuracy of the estimate significantly degrades if there is a burst a packet losses. Further, it is obvious that the accuracy of the estimate only depends on the deliveries of recent sensor readings, and readings in the distant past have negligible effect on the estimation accuracy. We thereby say that an information flow is \emph{credible} if its delivered packets enable the controller to make an accurate estimation.

The goal of this paper is to define and optimize the \emph{credibility} of an information flow that directly reflects the accuracy of the resulting estimate by the controller. To capture the aforementioned feature of real-time applications, we assume that the credibility of a real-time flow at a given point of time only depends on the packet deliveries in the \emph{window} of past $T$ intervals. Specifically, let $X_i(t)$ be the total number of timely-deliveries for flow $i$ in the first $t$ intervals. We then have $X_i(t)-X_i(t-1)=1$ if a packet is delivered to client $i$ in interval $t$, and $X_i(t)-X_i(t-1)=0$ if not. The number of timely-deliveries in the window of the last $T$ intervals can then be represented as $X_i(t)-X_i(t-T)$, and we assume that the credibility of flow $i$ at the end of interval $t$ only depends on the value of $X_i(t)-X_i(t-T)$.

We assume that, to make an accurate estimate, each client $i$ requires that there are at least $q_iT$ packets being delivered in the past $T$ intervals, i.e., $X_i(t)-X_i(t-T)\geq q_iT$. The value of $q_i$ depends on the context of the information flow. For example, a sensor monitoring a high-frequency signal requires a larger $q_i$ than one that is monitoring a low-frequency signal.

Due to the unreliable nature of wireless transmissions, it is obvious that it is not possible to satisfy the requirements of all clients at all time. When the AP fails to deliver $q_iT$ packets for a client $i$, then the estimation of current state of client $i$ becomes less accurate, and therefore we say that flow $i$ \emph{loses credibility}. 

We now formally define the measure of \emph{Loss-of-Credibility} (LoC). Suppose $X_i(t)-X_i(t-T)<q_iT$ for some $i$ and $t$. Recall that every transmission for client $i$ is successful with probability $p_i$. Hence, the AP would have needed to, on average, schedule $\frac{q_iT-(X_i(t)-X_i(t-T))}{p_i}$ more transmissions for client $i$ to make $X_i(t)-X_i(t-T)=q_iT$ and flow $i$ credible. We therefore define \emph{unbiased shortage} of client i at the end of interval $t$ as $\theta_i(t):=\max\{\frac{q_iT-(X_i(t)-X_i(t-T))}{p_i},0\}$. At the end of each interval $t$, each client $i$ suffers from a LoC of $C(\theta_i(t))$ based on its unbiased shortage, where $C(\cdot)$ is a strictly increasing, strictly convex, and differentiable function with $C(0)=0$ and $C'(0)=0$.

This paper aims to evaluate and minimize the long-time average total LoC of all clients in the system, which can be written as $\lim\limits_{\mathbb{T} \to \infty}\frac{{\sum \limits_{t = T+1}^{\mathbb{T}+T}{\sum\limits_{i = 1}^{N}{C(\theta_i(t))}}}}{\mathbb{T}} = \\ \lim\limits_{\mathbb{T} \to \infty}\frac{{\sum \limits_{t = T+1}^{\mathbb{T}+T}{\sum\limits_{i = 1}^{N}{C(\frac{q_i T}{p_i}-\frac{X_i(t)-X_i(t-T)}{p_i})}}}}{\mathbb{T}}$. 

\section{The Formulation of the Optimization Problem} \label{section:formulation}

In this section, we derive some fundamental properties about the minimization of total LoC. We then formulate an optimization problem.

Recall that $X_i(t)$ is the total number of timely-deliveries for client $i$ in the first $t$ intervals. Obviously, $\{X_i(1), X_i(2),\dots\}$ is a sequence of random variables whose distribution is determined by the employed packet scheduling policy. 
For simplicity, we only focus on ergodic scheduling policies in this paper. Thus, the random variable $\{X_i(t)-X_i(t-T)\}$ can be modeled by a positive recurrent Markov chain. By the law of large numbers, we can define $\bar{X}_i := \lim_{t \to \infty}{\frac{X_i(t)}{t}}$. Further, following the central limit theorem of Markov chains \cite{jones2004MarkovChainCLT}, $\hat{X}_i:=\lim_{\mathbb{T} \to \infty}{\frac{X_i(\mathbb{T})-\mathbb{T}\bar{X}_i}{\sqrt{\mathbb{T}}}}$ is a Guassian random variable with mean 0 and some finite variance, which we denote by $\sigma_i^2$, with $\sigma_i \geq 0$. Hence, we can approximate $X_i(t)-X_i(t-T)$ as a Gaussian random variable with mean $T\bar{X}_i$ and variance $T\sigma_i^2$. 
Let $\Phi(x)$ represents the cumulative distribution function of a random variable under standard normal distribution, then, under this approximation, we have that the CDF of $(X_i(t) - X_i(t-T))-T\bar{X}_i)$
is $\Phi(\frac{x}{\sqrt{\sigma_i^2 T}})$. 

The long-term average total LoC can now be re-written as below: 
\begin{align}   
    &\lim\limits_{\mathbb{T} \to \infty}\frac{{\sum \limits_{t = T+ 1}^{\mathbb{T}+T}{\sum\limits_{i = 1}^{N}{C(\frac{q_i T}{p_i}-\frac{X_i(t)-X_i(t-T)}{p_i})}}}}{\mathbb{T}}\notag \\
   =&\lim\limits_{\mathbb{T} \to \infty}{\sum _{i=1}^N{E[C(\frac{q_i T}{p_i}-\frac{X_i(\mathbb{T})-X_i(\mathbb{T} -T)}{p_i})]}} \notag \\
   \approx& \lim\limits_{\mathbb{T} \to \infty}\sum _{i=1}^N{E[C(\frac{q_i T}{p_i}-\frac{\sqrt{T}\hat{X}_i(T)+T\bar{X}_i}{p_i})]} \notag \\
    \approx &\sum_{i=1}^N{\int\limits_z{C(\sqrt{\frac{\sigma_i^2  T}{p_i^2}}z-\frac{(\bar{X}_i-q_i)T}{p_i})}}d\Phi(z).  \label{equation:objective}
\end{align}

Eq. (\ref{equation:objective}) has two sets of control variables: $[\bar{X}_i]$ and $[\sigma_i]$. Below, we derive the corresponding constraints of these two sets of variables.

We first derive the constraints on $[\bar{X}_i]$. Previous work \cite{hou2009theory} has shown that, under any work-conserving policy\footnote{A scheduling policy is called work-conserving if it always schedules a transmission when there is at least one packet available for transmission.}, we have, for all $t$,
\begin{equation}
E[\sum_{i=1}^N{\frac{{X}_i(t)-X_i(t-1)}{p_i}}]= \tau-I_{\{1,2,\dots, N\}}, \label{eq:x martingale}    
\end{equation}
and
\begin{equation}
E[\sum_{i\in S}{\frac{{X}_i(t)-X_i(t-1)}{p_i}}]\leq \tau-I_{S},    
\end{equation}
for any subset $S \subseteq \{1,2,\dots N\}$, where $I_s$ is called the \emph{idle time} and has been shown to be the same under all work-conversing policies. Therefore, we have
\begin{align}
\sum_{i=1}^N{\frac{\bar{X}_i}{p_i}}= \tau-I_{\{1,2,\dots, N\}}, \label{equation:constraint x}
\end{align}
and

\begin{equation}
\sum_{i \in S}{\frac{\bar{X}_i}{p_i}} \leq \tau-I_{S}, \forall S \subseteq \{1,2,\dots N\}. \label{equation:loose constraint x}
\end{equation} 
We further assume that, similar to the total resource pooling condition, the constraint $\sum_{i\in S}{\frac{\bar{X}_i}{p_i}}\leq \tau-I_S$ is not tight and can be ignored when $S$ is not $\{1,2,\dots N\}$. 

Now, we derive the constraint of $[\sigma_i]$. By (\ref{eq:x martingale}), the sequence of random variables
$\{\sum_{i=1}^N{\frac{{X}_i(t)}{p_i}}-t(\tau-I_{\{1,2,\dots, N\}})|t = 1,2,\dots\}$ is a martingale. By the martingale central limit theorem \cite{brown1971martingale}, 
$
\hat{X}_{TOT}:=\lim_{\mathbb{T} \to \infty}{\frac{\sum_{i=1}^N\frac{X_i(\mathbb{T})}{p_i}-\mathbb{T}(\tau-I_{\{1,2,\dots N\}})}{\sqrt{\mathbb{T}}}}=\lim_{\mathbb{T} \to \infty}{\frac{\sum_{i=1}^N\frac{X_i(\mathbb{T})}{p_i}-\mathbb{T}(\sum_{i=1}^N\frac{\bar{X}_i}{p_i})}{\sqrt{\mathbb{T}}}} $
is a Gaussian random variable with mean $0$, and its variance is 
\begin{align}
    \sigma_{TOT}^2:=&\lim_{\mathbb{T}\to\infty}\frac{1}{\mathbb{T}}[\sum_{t=1}^\mathbb{T}(\sum_{i=1}^N\frac{X_i(t)-X_i(t-1)}{p_i})^2]\notag \\ &-(\tau-I_{\{1,2,\dots N\}})^2,     
\end{align}whose value depends on the employed scheduling policy.

Recall that $\hat{X}_i:=\lim_{\mathbb{T} \to \infty}{\frac{X_i(\mathbb{T})-\mathbb{T}\bar{X}_i}{\sqrt{\mathbb{T}}}}$ is a Gaussian random variable with variance $\sigma_i^2$. Hence, we have $\hat{X}_{TOT}=\sum_{i=1}^N\frac{\hat{X}_i}{p_i}$, and
the variance of $\frac{\hat{X}_i}{p_i}$ is $(\frac{\sigma_i}{p_i})^2$. By Cauchy-Schwarz Inequality, we have:
\begin{align}
    &\big(\sum_{i=1}^N{\frac{\sigma_i}{p_i}}\big)^2 =\big(\sum\limits_{i=1}^N\sqrt{{Var(\frac{\hat{X}_i(t)}{p_i})}}\big)^2 \notag \\
    =  &\sum\limits_{i=1}^N{Var(\frac{\hat{X}_i(t)}{p_i})}\notag \\ &\quad +2\sum\limits_{l=1}^N\sum\limits_{m=l+1}^N{\sqrt{Var(\frac{\hat{X}_l(t)}{p_l})Var(\frac{\hat{X}_m(t)}{p_m})}}\notag \\
    \geq &\sum\limits_{i=1}^N{Var(\frac{\hat{X}_i(t)}{p_i})}+2\sum\limits_{l=1}^N\sum\limits_{m=l+1}^N{Cov(\frac{\hat{X}_l(t)}{p_l},\frac{\hat{X}_m(t)}{p_m})} \notag \\
    = &Var(\sum_{i=1}^N{\frac{\hat{X}_i(t)}{p_i}})=\sigma_{TOT}^2 \label{equation: constraint sigma},
\end{align}
where $Var(X)$ denotes the variance of $X$ and $Cov(X,Y)$ denotes the covariance. 

Although the value of $\sigma_{TOT}$ may be different for different scheduling policies, we first consider the special case of minimizing the total LoC when $\sigma_{TOT}$ is given and fixed. By (\ref{equation:objective}), (\ref{equation:constraint x}), and (\ref{equation: constraint sigma}), the optimization problem can be written as:
\begin{align}
    Min \quad &L=\sum_{i=1}^N{\int\limits_z{C(\sqrt{\frac{\sigma_i^2 T}{p_i^2}}z-\frac{(\bar{X}_i-q_i)T}{p_i})}}d\Phi(z) \label{equation: obj function} \\
    s.t. &\sum_{i=1}^N{\frac{\bar{X}_i}{p_i}} = \tau-I_{\{1,2,\dots N\}} \label{equation:first constriant} \\
    &\sum_{i=1}^N{\frac{\sigma_i}{p_i}} \geq \sigma_{TOT}. \label{equation:second constraint}
\end{align}
\begin{theorem} \label{theorem: opt solution}
Let $[\bar{X_i}^*]$ and $[\sigma_i^*]$ be the optimal solution to (\ref{equation: obj function}) -- (\ref{equation:second constraint}). Then $\bar{X_i}^* = (\frac{\tau - I_{\{1,2,\dots N\}}}{N}-\sum_{j=1}^N{\frac{q_j}{N p_j }}+\frac{q_i}{p_i})p_i$, and $\sigma_i^* = \frac{\sigma_{TOT}}{N}p_i$, for all $1\leq i\leq N$. 
\end{theorem}
\begin{proof}
Since $C(\cdot)$ is a convex function, we have: 
\begin{align*}
    L &= \sum_{i=1}^N{\int\limits_z{C(\sqrt{\frac{\sigma_i^2 T}{p_i^2}}z-\frac{(\bar{X}_i-q_i)T}{p_i})}d\Phi(z)}\notag \\ 
    &\geq N\int\limits_z{C(\frac{1}{N}\sum_{i=1}^N{(\sqrt{\frac{\sigma_i^2 T}{p_i^2}}z-\frac{(\bar{X}_i-q_i)T}{p_i}}))}d\Phi(z),
\end{align*}
with equality occurs when $\frac{\bar{X}_i^*}{p_i}-\frac{q_i}{p_i} = \frac{\bar{X}_j^*}{p_j}-\frac{q_j}{p_j}$ and $\frac{\sigma_i^*}{p_i} = \frac{\sigma_j^*}{p_j}$ for any $i, j \in \{1,2,\dots N\}$. By (\ref{equation:first constriant}) and (\ref{equation:second constraint}), we have $\bar{X_i}^* = (\frac{\tau - I_{\{1,2,\dots N\}}}{N}-\sum_{i =1}^N{\frac{q_i}{N p_i }}+\frac{q_i}{p_i})p_i$ and $\sigma_i^* = \frac{\sigma_{TOT}}{N}p_i$. 
\end{proof}

Theorem \ref{theorem: opt solution} establishes the optimal $\{\bar{X}_i\}$ and $\{\sigma_i\}$ that minimizes the total LoC when $\sigma_{TOT}$ is given and fixed. Obviously, smaller $\sigma_{TOT}$ leads to smaller total LoC. Therefore, we seek to solve the optimization problem below, which aims to minimizing $\sigma_{TOT}$ while satisfying the results of Theorem~\ref{theorem: opt solution}: 
\begin{align}
    Min \quad&\sigma_{TOT}^2:=\lim_{\mathbb{T}\to\infty}\frac{1}{\mathbb{T}}[\sum_{t=1}^\mathbb{T}(\sum_{i=1}^N\frac{X_i(t)-X_i(t-1)}{p_i})^2]\notag \\&-(\tau-I_{\{1,2,\dots N\}})^2 \label{equation:second objective}\\
    s.t. &\bar{X}_i = \bar{X}_i^*, \forall 1\leq i\leq N \label{equation:x opt solution}\\
    &\sigma_i = \frac{\sigma_{TOT}}{N}p_i, \quad \forall 1\leq i\leq N,\label{equation: sigma opt solution}
\end{align} 
where $\bar{X}_i^*:= (\frac{\tau-I_{\{1,2,\dots N\}}}{N}-\sum_{j=1}^N{\frac{q_j}{p_j N}}+\frac{q_i}{p_i})p_i$.

We note that the problem (\ref{equation:second objective}) -- (\ref{equation: sigma opt solution}) involves both a constraint on the average of $X_i(t)$ (\ref{equation:x opt solution}) and a constraint on the variance of $X_i(t)$ (\ref{equation: sigma opt solution}) for each $i$. Most existing studies on network utility maximization (NUM) problem only addresses constraints on the average of decision variables, and therefore cannot be applied to solve (\ref{equation:second objective}) -- (\ref{equation: sigma opt solution}). In fact, no stationary randomized policies can optimally solve (\ref{equation:second objective}) -- (\ref{equation: sigma opt solution}). In the following sections, we will establish the surprising result that there exists a simple online scheduling policy that is near-optimal for the problem (\ref{equation:second objective}) -- (\ref{equation: sigma opt solution}). 
\section{An online scheduling policy} \label{section: policy}


In this section, we propose a simple online scheduling policy for the problem (\ref{equation:second objective}) -- (\ref{equation: sigma opt solution}). We first provide a brief outline of the construction of our algorithm. First, we remove the constraint on variance (\ref{equation: sigma opt solution}) and focus on the following optimization problem:

\begin{align}
    Min \quad &\lim_{\mathbb{T}\to\infty}\frac{1}{\mathbb{T}}[\sum_{t=1}^\mathbb{T}(\sum_{i=1}^N\frac{X_i(t)-X_i(t-1)}{p_i})^2]\notag \\&\quad -(\tau-I_{\{1,2,\dots N\}})^2 \label{equation:third objective}\\
    s.t. &\bar{X}_i = \bar{X}_i^*,  \forall 1\leq i\leq N. \label{equation:third constraint for x}
\end{align}
Obviously, this optimization problem is a lower bound to the original problem (\ref{equation:second objective}) -- (\ref{equation: sigma opt solution}). It is also a standard NUM problem that only involves a constraint on the average of $X_i(t)$ for each $i$. We can therefore derive a near-optimal online scheduling algorithm using the Drift-Plus-Penalty approach \cite{neely2010stochastic}. We further demonstrate the surprising result that, due to the specific choice of our Lyapunov function, our algorithm also satisfies the constraint on variance (\ref{equation: sigma opt solution}). Therefore, our algorithm is near-optimal to the original problem (\ref{equation:second objective}) -- (\ref{equation: sigma opt solution}).

We now introduce some notations that are necessary for the design and analysis of our algorithm. Let $d_i(t):=\frac{\bar{X}_i^* t}{p_i}-\frac{X_i(t)}{p_i}$ be the \emph{deficit} of client $i$ in interval $t$. Obviously, we have $\bar{X}_i := \lim_{t \to \infty}{\frac{X_i(t)}{t}}=\bar{X}_i^*$ if and only if $\lim_{t \to \infty}\frac{d_i(t)}{t}=0$. We also define $\Delta d_i(t) := d_i(t+1)-d_i(t)=\frac{\bar{X}_i^* }{p_i}-\frac{X_i(t+1)-X_i(t)}{p_i}$ and $D(t) := \frac{\sum_{i=1}^N d_i(t)}{N}$.


We consider the Lyapunov function $L(t) = \frac{1}{2}\sum_{i =1}^N[d_i(t)-D(t)]^2$. The drift of the Laypunov function is $\Delta L(t) := E[L(t+1)-L(t)|[d_i(t)] ]$. 


Given $[d_i(t)]$, we have, under any scheduling policy,
\begin{align}
    &\Delta {L(t)}= E[L(t+1)-L(t)]\notag\\
    =& E\big[\frac{1}{2}\sum_{i=1}^N\big( d_i(t+1)-D(t+1)\big)^2 \notag \\ &\quad - \frac{1}{2}\sum_{i =1}^N\big(d_i(t)-D(t)\big)^2\big] \notag\\
    =& E\big[\frac{1}{2}\sum_{i =1}^N\big(d_i(t)-D(t)+\Delta d_i(t)-\frac{\sum_{i=1}^N\Delta d_i(t)}{N}\big)^2\big]\notag \\ & - E\big[\frac{1}{2}\sum_{i =1}^N\big(d_i(t)-D(t)\big)^2\big] \notag\\
    =& E\big[\frac{1}{2}\sum_{i=1}^N\big(\Delta d_i(t) -\frac{\sum_{i=1}^N\Delta d_i(t)}{N}\big)^2\big]\notag \\ & +\sum_{i=1}^NE\big[\Delta d_i(t)\big]\big(d_i(t)-D(t)\big)\notag \\ & - E\big[\frac{\sum_{i=1}^N{\Delta d_i(t)}}{N}\big]\sum_{i =1}^N\big(d_i(t)-D(t)\big)\notag\\
    \leq& \beta + \sum_{i=1}^NE\big[\Delta d_i(t)\big]\big(d_i(t)-D(t)\big), \label{equation:drift expand}
\end{align} 
where $\beta$ is a bounded positive number. The last inequality holds since $\Delta d_i(t)$ is bounded by $\frac{\bar{X}_i^*-1}{p_i}\leq \Delta d_i(t)\leq \frac{\bar{X}_i^*}{p_i}$ and $\sum_{i =1}^N d_i(t)=ND(t)$.

Our scheduling algorithm is based on the Drift-Plus-Penalty approach \cite{neely2010stochastic}. Let 
\begin{align}
B(t):=&\sum_{i=1}^NE\big[\Delta d_i(t)\big]\big(d_i(t)-D(t)\big)\notag \\
&+\epsilon E[(\sum_{i=1}^N\frac{X_i(t+1)-X_i(t)}{p_i})^2], \label{equation:definition B}
\end{align}
where $\epsilon$ is a positive number whose value can be arbitrary determined by the system designer. We then have
\begin{align}
    \Delta L(t)+\epsilon E[(\sum_{i=1}^N\frac{X_i(t+1)-X_i(t)}{p_i})^2] \leq \beta + B(t). \label{equation: drift plus penalty}
\end{align}

We aim to design an online scheduling algorithm that minimizes $B(t)$. Note that the value of $B(t)$ depends on the scheduling decisions on all time slots within the interval $t$, which consists of $\tau$ time slots. Minimizing an objective function over a finite horizon of $\tau$ time slots typically requires the usage of dynamic programming. However, we will show that there exists a simple online scheduling algorithm that minimizes $B(t)$. 

Our algorithm is called the \emph{Minimum-Drift-and-Variance-First} (MDVF) policy. Under the MDVF policy, the AP calculates the value of $r_i(t) := \epsilon \frac{1}{p_i} - d_i(t)$ at the beginning of each interval $t$. In each time slot within the interval, the AP finds the undelivered packet with the smallest $r_i(t)$ and transmits that packet, as long as there is at least one packet to be transmitted.





\begin{lemma} \label{lemma: drift and sigma expand}
The MDVF policy minimizes $B(t)$.
\end{lemma}
\begin{proof}
We prove this lemma by induction. First, we consider the optimal scheduling decision in the last time slot of the interval. At this time, some packets have already been delivered in the previous $\tau-1$ slots, and we use $V$ to denote the set of clients whose packets have already been delivered. As this is the last time slot of the interval, the scheduling decision of the AP only consists of choosing one client $u\notin V$ and transmitting its packet. Given $V$ and $u$, we will calculate the value of $\sum_{i=1}^NE\big[\Delta d_i(t)\big]\big(d_i(t)-D(t)\big)+\epsilon E[(\sum_{i=1}^N\frac{X_i(t+1)-X_i(t)}{p_i})^2]$.

For this chosen client $u$, its packet will be delivered, that is, $X_u(t+1)-X_u(t)=1$, with probability $p_u$, and $X_u(t+1)-X_u(t)=0$, with probability $1-p_u$. Hence, we have $E[\Delta d_u(t)] = \frac{\bar{X}_u - p_u}{p_u}$.

On the other hand, for each client $i\in V$, its packet has already been delivered. We have $X_v(t)-X_v(t-1)=1$ and $E[\Delta d_i(t)] = \frac{\bar{X}_i - 1}{p_i}$.

Finally, for each client $i\notin V\cup \{u\}$, its packet will not be delivered, and we have $X_i(t)-X_i(t-1)=0$ and $E[\Delta d_i(t)] = \frac{\bar{X}_i}{p_i}$.

We now have, given $V$ and $u$,
\begin{align}
    &\sum_{i=1}^NE\big[\Delta d_i(t)\big][d_i(t)-D(t)]\notag \\ &\quad +\epsilon E[(\sum_{i=1}^N\frac{X_i(t+1)-X_i(t)}{p_i})^2]\notag \\
    &=\frac{\bar{X}_u - p_u}{p_u}[d_u(t)-D(t)] + \sum_{i\in V}\frac{\bar{X}_i - 1}{p_i}[d_i(t)-D(t)] \notag\\
    &\quad +\sum_{i\notin V\cup \{u\}}\frac{\bar{X}_i}{p_i}[d_i(t)-D(t)]\notag\\ 
    &\quad +\epsilon [p_u\big(\sum_{i\in V}\frac{1}{p_i}+\frac{1}{p_u}\big)^2+(1-p_u)\big(\sum_{i\in V}\frac{1}{p_i}\big)^2]\notag\\
    &= \epsilon\frac{1}{p_u}-d_u(t) +\lambda(V), \label{equation: epsilon sigma expand}
\end{align}
where $\lambda(V):=D(t)+\sum_{i=1}^N \frac{\bar{X}_i}{p_i} [d_i(t) - D(t)]\\ - \sum_{i \in V} \frac{1}{p_i} [d_i(t) - D(t)]+\epsilon [\big(\sum_{i\in V}\frac{1}{p_i}\big)^2+2\big(\sum_{i\in V}\frac{1}{p_i}\big)]$ is the same regardless of the choice of $u$. Therefore, it is clear that an optimal scheduling algorithm that minimizes $B(t)$ will schedule the undelivered packet $u$ with the smallest $\epsilon\frac{1}{p_u}-d_u(t)$ in the last time slot.

Now, assume that, starting from the $(s+1)$-th time slot in an interval, scheduling the undelivered packet with the smallest $\epsilon\frac{1}{p_u}-d_u(t)$ in each of the remaining time slot is optimal. We will show that, even in the $s$-th time slot, scheduling the undelivered packet with the smallest $\epsilon\frac{1}{p_u}-d_u(t)$ is optimal.

We prove this claim by contradiction. Let $u^*$ be the undelivered packet with the smallest $\epsilon\frac{1}{p_u}-d_u(t)$ in time slot $s$. If the claim is false, then the optimal scheduling algorithm, which we denote by $\mathbb{A}$, would schedule another undelivered packet $u'\neq u^*$ in time slot $s$, and the value of $B(t)$ under $\mathbb{A}$ is strictly smaller than any policy that schedules $u^*$ in the $s$-th time slot. By the induction hypothesis, $\mathbb{A}$ begins to schedule the undelivered packet with the smallest $\epsilon\frac{1}{p_u}-d_u(t)$ starting from the $(s+1)$-th time slot. As $u^*$ is not scheduled by $\mathbb{A}$ is the $s$-th time slot, $\mathbb{A}$ needs to schedule $u^*$ in the $(s+1)$-th time slot. In summary, $\mathbb{A}$ schedules $u'$ in the $s$-th time slot, and $u^*$ in the $(s+1)$-th time slot.

Now, we can construct another algorithm $\mathbb{B}$ by simply swapping the transmissions in the $s$-th time slot and the $(s+1)$-th time slot. In other words, $\mathbb{B}$ schedules $u^*$ in the $s$-th time slot, $u'$ in the $(s+1)$-th time slot, and then follows $\mathbb{A}$ starting from the $(s+2)$-th time slot. Obviously, the value of $B(t)$ under $\mathbb{A}$ and $\mathbb{B}$ is the same, which results in a contradiction.

We have established that, even in the $s$-th time slot, scheduling the undelivered packet with the smallest $\epsilon\frac{1}{p_u}-d_u(t)$ is optimal. By induction, scheduling the undelivered packet with the smallest $\epsilon\frac{1}{p_u}-d_u(t)$ in each time slot is optimal, and MDVF minimizes $B(t)$.
\end{proof}

\section{Performance Analysis of the MDVF policy}\label{section:performance}

We now study the performance of the MDVF policy. We will demonstrate the surprising result that the MDVF policy satisfies both constraints on mean (\ref{equation:x opt solution}) and variance (\ref{equation: sigma opt solution}), and the value of $\sigma_{TOT}^2$ under the MDVF policy can be made arbitrary close to a lower bound. Throughout this section, we use $\cdot|\eta$ to denote the value of $\cdot$ under a scheduling policy $\eta$. For example, $\Delta L(t)|$MDVF denotes the value of $\Delta L(t)$ under the MDVF policy.



We first establish the following property.

\begin{theorem} \label{theorem:positive recurrent Markov}
Under the MDVF policy, the Markov process with state vector $\{d_i(t)-D(t)\}$ is positive recurrent.
\end{theorem}
\begin{proof}
We prove this theorem by establishing an upper bound of $\Delta L(t)|$MDVF. To simplify notations, we let $\Omega$ be the policy that schedules the undelivered packet with the maximum value of $d_i(t)$. We also sort all clients such that $d_1(t)\geq d_2(t)\geq\dots \geq d_N(t)$. Then $\Omega$ will only transmit a packet for client $i$ if, for each $j<i$, the packet for flow $j$ has already been delivered. This is equivalent to the largest-debt-first policy in \cite{hou2009theory}, and we have, for all $1\leq j\leq N$:
\begin{align}
\sum_{i=1}^jE[\Delta d_i(t)]|\Omega=&\sum_{i=1}^{j}\frac{\bar{X}_i^*}{p_i} - E[\sum_{i=1}^j\frac{X_i(t+1)-X_i(t)}{p_i}]|\Omega\notag\\
=&\sum_{i=1}^{j}\frac{\bar{X}_i^*}{p_i}-(\tau-I_{\{1,2,\dots,j\}}).
\end{align}
By (\ref{equation:constraint x}), we have $\sum_{i=1}^NE[\Delta d_i(t)]|\Omega=0$. Further, as we assume that (\ref{equation:loose constraint x}) is not tight when $S\neq \{1,2,\dots,N\}$, there exists a positive number $\delta>0$ such that $\sum_{i=1}^jE[\Delta d_i(t)]|\Omega \leq-\delta$ for all $1\leq j\leq N-1$. We now have
\begin{align}
    & \sum_{i=1}^NE[\Delta d_i(t)] \big(d_i(t)-D(t)\big)|\Omega \notag\\
    =& \sum_{i=1}^NE[\Delta d_i(t)] \big(d_i(t)-d_{i+1}(t)+d_{i+1}(t)\notag \\
    & -d_{i+2}(t)+\dots    -d_{N}(t)+d_{N}(t)-D(t)\big)|\Omega \notag\\
    =& \sum_{i=1}^NE[\Delta d_i(t)] \big(d_{N}(t)-D(t)\big)|\Omega \notag \\ 
    & + \sum_{i=1}^j{\sum_{j=1}^{N-1}E[\Delta d_i(t)] \big(d_j(t)-d_{j+1}(t)\big)}|\Omega \notag\\
     \leq& -\delta \sum_{j=1}^{N-1}\big(d_j(t)-d_{j+1}(t)\big) 
    = - \delta \big(d_{1}(t)-d_{N}(t)\big).\label{equation:omega policy}
\end{align}

Next, we study $\Delta L(t)|$MDVF. By Lemma \ref{lemma: drift and sigma expand}, the MDVF policy minimizes $B(t)$. Hence, we have
\begin{align}
    &\Delta L(t)|\mbox{MDVF}+\epsilon E[(\sum_{i=1}^N\frac{X_i(t+1)-X_i(t)}{p_i})^2]|\mbox{MDVF}  \notag\\
    \leq& \beta + B(t)|\mbox{MDVF} \hspace{30pt}\mbox{\big(By (\ref{equation: drift plus penalty}) \big)} \notag\\ 
    \leq& \beta + B(t)|\Omega \notag\\  
     \leq& \beta +\epsilon E[(\sum_{i=1}^N\frac{X_i(t+1)-X_i(t)}{p_i})^2]|\Omega  \notag \\
    &- \delta\big(d_{1}(t)-d_{N}(t)\big) \hspace{30pt}\mbox{\big(By (\ref{equation:definition B}) and (\ref{equation:omega policy})\big)}
\end{align} 
Since $0\leq X_i(t+1)-X_i(t)\leq 1$, there exists some constant $M$ such that
\begin{align}
    \Delta L(t)|\mbox{MDVF} &\leq -  \delta\big(d_{1}(t)-d_{N}(t)\big)+ M. \label{equation:lyapunov result with beta and M}
\end{align}

Recall that we have sorted all clients such that $d_1(t)\\ \geq d_2(t)\geq \dots$. Hence, $\big(d_1(t)-d_N(t)\big)\geq0$ and  $\big(d_1(t)\\-d_N(t)\big)\geq|d_i(t)-D(t)|$, for all $i$. We have
\begin{align*}
    &\Delta L(t)|\mbox{MDVF} < -\delta, \mbox{if $|d_i(t)-D(t)|>\frac{M}{\delta}+1$, for some $i$},
\end{align*}
and
\begin{align}
    {\Delta L(t)|\mbox{MDVF}} \leq M, \mbox{ otherwise.}
\end{align}
By the Foster-Lyapunov Theorem, the Markov process with state vector $\{d_i(t)-D(t)\}$ is positive recurrent.
\end{proof}

Now we are able to show that the MDVF policy satisfies both constraints (\ref{equation:x opt solution}) and (\ref{equation: sigma opt solution}).
\begin{corollary}
$\bar{X}_i|\mbox{MDVF}=\bar{X}_i^*$ and $\sigma_i|\mbox{MDVF}=\\ \frac{\sigma_{TOT}|MDVF}{N}p_i$, $\forall i$.
\end{corollary}
\begin{proof}
Recall that $d_i(t):=\frac{\bar{X}_i^* t}{p_i}-\frac{X_i(t)}{p_i}$ and $D(t) :=\\ \frac{\sum_{i=1}^N d_i(t)}{N}$. By (\ref{equation:constraint x}), we have:
\begin{align}
     &\lim_{\mathbb{T} \to \infty} \frac{D(\mathbb{T})|\mbox{MDVF}}{\mathbb{T}} = \lim_{\mathbb{T} \to \infty}\frac{\sum_{i=1}^N d_i(\mathbb{T})|\mbox{MDVF}}{N \mathbb{T}}\notag \\
     =& \frac{1}{N} \sum_{i=1}^N \lim_{\mathbb{T} \to \infty} \frac{ \mathbb{T}\bar{X}_i^*-X_i(\mathbb{T})|\mbox{MDVF}}{p_i \mathbb{T}}\notag \\
     =&  \frac{1}{N} \sum_{i=1}^N \frac{\bar{X}_i^*}{p_i}-\frac{1}{N}\sum_{i=1}^N\frac{\bar{X}_i(\mathbb{T})|\mbox{MDVF}}{p_i} \notag \\
     =&\frac{\tau-I_{\{1,2,\dots N\}}}{N}-\frac{\tau-I_{\{1,2,\dots N\}}}{N}=0. \label{equation:lim D over T} 
\end{align}

By Theorem~\ref{theorem:positive recurrent Markov}, the vector $\{d_i(t)-D(t)\}|$MDVF converges to a steady state distribution as $t\to\infty$. Hence, both $\lim_{\mathbb{T} \to \infty}{\frac{d_i(\mathbb{T})-D(\mathbb{T})}{{\mathbb{T}}}}|$MDVF and $\lim_{\mathbb{T} \to \infty}{\frac{d_i(\mathbb{T})-D(\mathbb{T})}{\sqrt{\mathbb{T}}}}|$MDVF converge to $0$ in probability. We then have
\begin{align}
    &\lim_{\mathbb{T} \to \infty} \frac{d_i(\mathbb{T})|\mbox{MDVF}}{\mathbb{T}} =\frac{\bar{X}_i^*}{p_i}-\frac{\bar{X}_i|\mbox{MDVF}}{p_i}\notag \\ 
    =&\lim_{\mathbb{T} \to \infty} \frac{D(\mathbb{T})|\mbox{MDVF}}{\mathbb{T}} = 0,
\end{align}and hence $\bar{X}_i|$MDVF$=\bar{X}_i^*  $.

Next, we study $\sigma_i|$MDVF. Recall that $\sigma_i^2$ is the variance of\\ $\hat{X}_i:=\lim_{\mathbb{T} \to \infty}{\frac{X_i(\mathbb{T})-\mathbb{T}\bar{X}_i}{\sqrt{\mathbb{T}}}}$. We then have:
\begin{align}
    \lim_{\mathbb{T} \to \infty}{\frac{d_i(\mathbb{T})|\mbox{MDVF}}{\sqrt{\mathbb{T}}}} &= \lim_{\mathbb{T} \to \infty}{\frac{\mathbb{T}\bar{X}_i^*-X_i(\mathbb{T})|\mbox{MDVF}}{p_i\sqrt{\mathbb{T}}}}\notag \\
    &= -\frac{\Hat{X}_i|\mbox{MDVF}}{p_i} ,\notag
\end{align} since $\bar{X}_i|$MDVF$=\bar{X}_i^*$. This shows that the variance of \\$\lim_{\mathbb{T} \to \infty}{\frac{d_i(\mathbb{T})|\mbox{MDVF}}{\sqrt{\mathbb{T}}}}$ is $\frac{\sigma_i^2|\mbox{MDVF}}{p_i^2}$.

Also, recall that $\sigma_{TOT}^2$ is the variance of $\hat{X}_{TOT}=\\ \sum_{i=1}^N\frac{\hat{X}_i}{p_i}$. We have
\begin{align}
    &\lim_{\mathbb{T} \to \infty}{\frac{D(\mathbb{T})|\mbox{MDVF}}{\sqrt{\mathbb{T}}}}=\lim_{\mathbb{T} \to \infty}{\frac{\sum_{i=1}^N{d_i(\mathbb{T})|\mbox{MDVF}}}{N\sqrt{\mathbb{T}}}}\notag \\
    =& \lim_{\mathbb{T} \to \infty}{\sum_{i=1}^N\frac{\mathbb{T}\bar{X}_i^*
    -X_i(\mathbb{T})|\mbox{MDVF}}{Np_i \sqrt{\mathbb{T}}}}
    = -\sum_{i=1}^N\frac{\Hat{X}_i|\mbox{MDVF}}{Np_i},\notag
\end{align}
and the variance of $\lim_{\mathbb{T} \to \infty}{\frac{D(\mathbb{T})|\mbox{MDVF}}{\sqrt{\mathbb{T}}}}$ is $\frac{\sigma_{TOT}^2|\mbox{MDVF}}{N^2}$. As $\lim_{\mathbb{T} \to \infty}{\frac{d_i(\mathbb{T})-D(\mathbb{T})}{\sqrt{\mathbb{T}}}}|$MDVF converges to $0$ in probability, we have $\sigma_i|\mbox{MDVF}=\frac{\sigma_{TOT}|\mbox{MDVF}}{N}p_i$.
\end{proof}

We have shown that the MDVF policy satisfies both constraints (\ref{equation:x opt solution}) and (\ref{equation: sigma opt solution}). We now show that the value of $\sigma_{TOT}^2|\mbox{MDVF}$ can be made arbitrarily close to a theoretical lower bound.

Consider the problem (\ref{equation:third objective}) -- (\ref{equation:third constraint for x}), which ignores the constraint on variance (\ref{equation: sigma opt solution}). Since this problem only involves a constraint on mean, there exists a stationary randomized policy that is optimal, which we denote by $\omega$. Obviously, $\sigma_{TOT}^2|\omega$ is a lower bound of the problem (\ref{equation:second objective}) -- (\ref{equation: sigma opt solution}). We have the following theorem.
\begin{theorem} \label{theorem:sigma approaching}
$\sigma_{TOT}^2|\mbox{MDVF}\leq \sigma_{TOT}^2|\omega + \frac{\beta}{\epsilon}$.
\end{theorem}
\begin{proof}
Since $\omega$ is a stationary randomized policy that satisfies (\ref{equation:third constraint for x}), we have $E[\Delta d_i(t)]|\omega = 0$, for all $i$ and $t$. By (\ref{equation:definition B}), we have
\begin{align*}
B(t)|\omega=\epsilon E[(\sum_{i=1}^N\frac{X_i(t+1)-X_i(t)}{p_i})^2]|\omega=\epsilon\sigma_{TOT}^2|\omega.
\end{align*}
Now, recall that the MDVF policy minimizes $B(t)$. Hence, for every $t$, we have
\begin{align*}
    &{\Delta L(t)|\mbox{MDVF}} +\epsilon E[(\sum_{i=1}^N\frac{X_i(t)-X_i(t-1)}{p_i})^2]|\mbox{MDVF} \notag \\ 
    \leq& B(t)|\mbox{MDVF} + \beta \notag \\
    \leq& B(t)|\omega + \beta =\epsilon\sigma_{TOT}^2|\omega+\beta\notag.
\end{align*}
Summing the above inequality over $t=1$ to $t=\mathbb{T}$, and then divide both sides by $\mathbb{T}$ yields
\begin{align}
    &\frac{E[L(\mathbb{T}+1)]-E[L(0)]}{\mathbb{T}}|\mbox{MDVF} +\epsilon\sigma_{TOT}^2|\mbox{MDVF}\notag\\
    \leq& \epsilon\sigma_{TOT}^2|\omega+\beta.
\end{align}
By Theorem~\ref{theorem:positive recurrent Markov}, we have $\lim_{\mathbb{T}\to\infty}\frac{E[L(\mathbb{T}+1)]-E[L(0)]}{\mathbb{T}}|$MDVF\\$=0$, and hence $\sigma_{TOT}^2|\mbox{MDVF}\leq \sigma_{TOT}^2|\omega + \frac{\beta}{\epsilon}$.
\end{proof}

We note that Theorem~\ref{theorem:sigma approaching} holds for all $\epsilon$, which is a constant that can be arbitrarily chosen by the system designer. By choosing a large $\epsilon$, one can make $\sigma_{TOT}^2|\mbox{MDVF}$ arbitrarily close to the lower bound $\sigma_{TOT}^2|\omega$.

\section{Simulation Results} \label{sec:simulation}

We present our simulation results in this section. We have implemented and tested our policy and two other state-of-the-art policies in ns-2. All simulations are conducted using the 802.11 MAC protocol with 54Mbps data rate. Simulations show that the time needed to transmit a packet and to receive an ACK is about $0.5ms$. The duration of an interval is chosen to be $10ms$, or, equivalently, 20 time slots. The LoC function is chosen to be $C(\theta)=\theta^2$ when $\theta>0$. All results presented in this section are the average of 1000 runs.


We compare our MDVF policy against two other policies. The first policy is the largest debt first (LDF) policy in \cite{hou2009theory, hou2014scheduling}. In each interval $t$, the LDF policy sorts all clients in descending order of $q_it-X_i(t)$, and transmit packets according to this ordering. It has been shown that LDF guarantees to deliver a long-term average timely-throughput of $q_i$ to each client $i$, as long as it is feasible to do so. The second policy is a Max-Weight type of policy that aims to reduce the total age-of-information (AoI) in the network while guaranteeing some average timely-throughput policy \cite{kadota2018optimizing}. We call this policy MW-AoI. Although the problem of minimizing AoI remains an open problem, it has been shown that the MW-AoI policy is 4-optimal in terms of AoI.

As for the network topology, we consider two different settings. In the first setting, there are 12 wireless clients. The channel reliability of client $i$ is set to be $p_i=0.9-0.05i$. We set $q_i=0.85$ for the first 6 clients and $q_i=0.75$ for the last 6 clients. We call this setting the \emph{high-timely-throughput system}. In the second setting, there are 18 clients with $p_i=1-0.05i$. We set $q_i=0.5$ for the first 9 clients and $q_i=0.35$ for the last 9 clients. We call this setting the \emph{low-timely-throughput system.}

For each simulation run, we record the total LoC incurred in the past second. Simulation results of the two systems are shown in Fig.~\ref{fig:simulation:P}. Simulation results clearly show that our MDVF policy achieves the smallest LoC for both systems. A very surprising result is that the MW-AoI policy has the highest LoC. The reason is that the MW-AoI policy focuses on optimizing AoI, which only depends on the time of the most recent packet delivery. However, most estimation techniques require more than the most recent data to make an accurate estimation. Even basic techniques like linear extrapolation needs at least two data points to make an estimate. This simulation result highlights that AoI may fail to completely capture the accuracy of estimation. On the other hand, the LDF policy only aims to optimize the long-term average timely-throughputs and ignores temporal variance. This leads it to also have suboptimal total LoC.

\begin{figure}[htbp]
\centering
\subfigure[The high-timely-throughput system]{
\label{fig:HF-P} 
\includegraphics[width=2.5in]{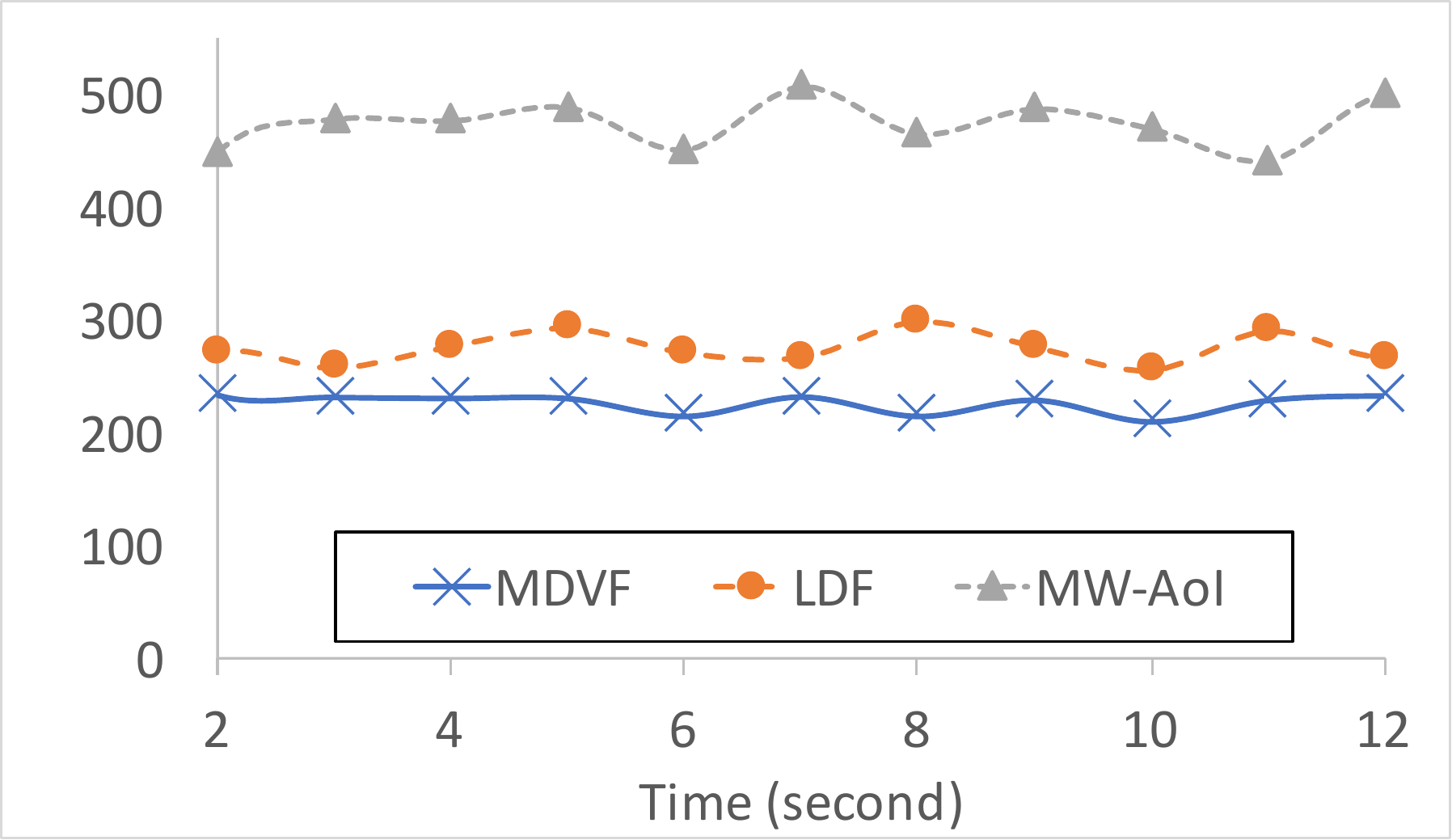}}
\hfill
\hspace{0.01\linewidth} 
\subfigure[The low-timely-throughput system]{
\label{fig:LF-P} 
\includegraphics[width=2.5in]{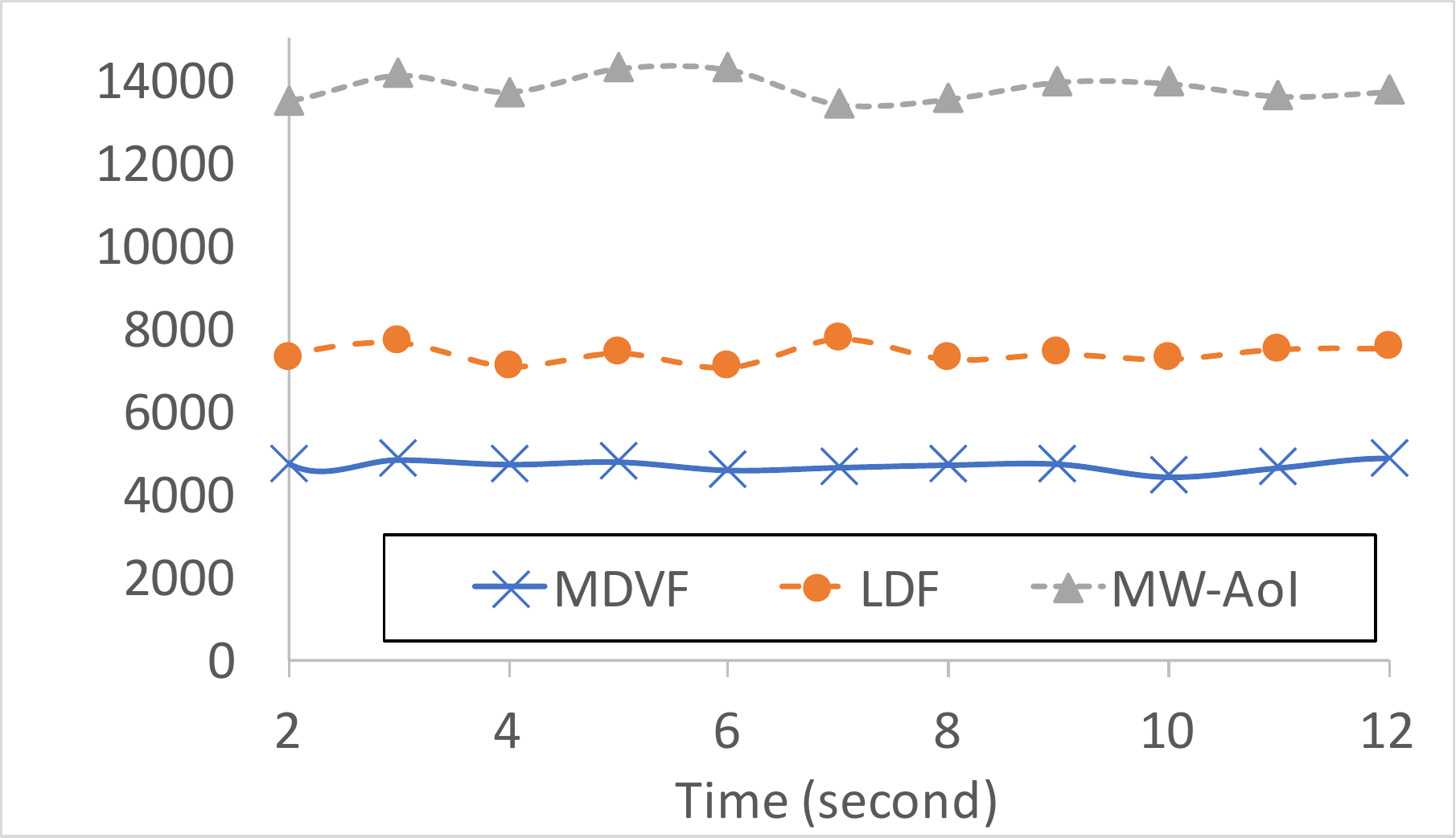}}
\caption{The total LoC in the past second.}\label{fig:simulation:P}
\end{figure}

\section{Related Work} \label{section: related work}

Real-time wireless networks have gained a lot of research interests. Hou, Borkar, and Kumar \cite{hou2009theory} have proposed a frame-based model to describe delay requirements of real-time flows. Under this model, the performance of each flow is determined by its timely-throughput, which is the long-term average number of timely deliveries. Jaramillo, Srikand, and Ying \cite{Jaramillo2011QOSinRealTimeTraffic} have studied wireless flows with heterogeneous delay and timely-throughput requirements. Kang et. al.\cite{kang2013LDFperformanceOnRealTime} have studied the performance of timely-throughputs in ad hoc wireless networks with stochastic packet arrivals. Meko and Seid \cite{Meko2015PacketLostOPT} have proposed a randomized scheduling algorithm for real-time flows. Zhang et. al. \cite{G.Zhang2015TimelyThruputHCN} have studied timely-throughputs in heterogeneous cellular networks with mobile nodes. Lashgari and Avestimehr \cite{Lashgari2013HeterNetworkTotalThruput} have looked for the additive gap of maximal timely throughput in a relaxed problem under the time-varying channel states. However, all these studies focus on the long-term average timely-throughput of each flow. As demonstrated in this paper, the temporal variance of timely-throughput can have significant impact on the credibility of an information flow. Singh, Hou, and Kumar \cite{singh2014fluctuation} have studied the fluctuation of timely-throughput, but its results only hold for a limiting scaled workloads. Hou \cite{HOU2016ShortTermPerformance} has proposed a scheduling policy to optimize the short-term performance of real-time flows, but the policy only applies to wireless networks where all links have the same quality.



Age-of-information (AoI) is another metric that aims to capture the short-term performance of information flows that has gained a lot of research interests \cite{Kaul2011AoIinVehicular, Modiano2017AoIinMHopsNetwork, kadota2018optimizing, Modiano2018AoIBroadcast, Modiano2018PerfectChannelInfoAoIOptimal, He2016AoILinkScheduling}. AoI is defined to capture the time of the most recent packet delivery. As shown in this paper, it may not be sufficient to capture the accuracy of estimation algorithms, which typically need multiple data points to make an estimation.

\section{Conclusion} \label{section: conclusion}

We have studied the problem of minimizing the total Loss-of-Credibility (LoC) in real-time wireless networks, where the LoC of each flow only depends on the timely deliveries in a window of the recent past. We have shown that, unlike most existing network utility maximization (NUM) problem, the problem of minimizing total LoC requires the precise control of the temporal variance of timely deliveries. To solve this problem, we have proposed a simple online algorithm called the MDVF policy, and have proved that the MDVF policy is near-optimal. Simulation results have further demonstrated that the MDVF policy outperforms other state-of-the-art policies.

\bibliographystyle{IEEEtran}
\bibliography{reference}

\begin{thebibliography}{10}
\providecommand{\url}[1]{#1}
\csname url@samestyle\endcsname
\providecommand{\newblock}{\relax}
\providecommand{\bibinfo}[2]{#2}
\providecommand{\BIBentrySTDinterwordspacing}{\spaceskip=0pt\relax}
\providecommand{\BIBentryALTinterwordstretchfactor}{4}
\providecommand{\BIBentryALTinterwordspacing}{\spaceskip=\fontdimen2\font plus
\BIBentryALTinterwordstretchfactor\fontdimen3\font minus
  \fontdimen4\font\relax}
\providecommand{\BIBforeignlanguage}[2]{{%
\expandafter\ifx\csname l@#1\endcsname\relax
\typeout{** WARNING: IEEEtran.bst: No hyphenation pattern has been}%
\typeout{** loaded for the language `#1'. Using the pattern for}%
\typeout{** the default language instead.}%
\else
\language=\csname l@#1\endcsname
\fi
#2}}
\providecommand{\BIBdecl}{\relax}
\BIBdecl

\bibitem{HOU2016ShortTermPerformance}
I.~Hou, ``On the modeling and optimization of short-term performance for
  real-time wireless networks,'' in \emph{IEEE INFOCOM 2016 - The 35th Annual
  IEEE International Conference on Computer Communications}, April 2016, pp.
  1--9.

\bibitem{jones2004MarkovChainCLT}
\BIBentryALTinterwordspacing
G.~L. Jones, ``On the markov chain central limit theorem,'' \emph{Probab.
  Surveys}, vol.~1, pp. 299--320, 2004. [Online]. Available:
  \url{https://doi.org/10.1214/154957804100000051}
\BIBentrySTDinterwordspacing

\bibitem{hou2009theory}
I.-H. Hou, V.~Borkar, and P.~R. Kumar, ``A theory of qos for wireless,'' in
  \emph{INFOCOM 2009, IEEE}, 2009, pp. 486--494.

\bibitem{brown1971martingale}
B.~M. Brown, ``Martingale central limit theorems,'' \emph{The Annals of
  Mathematical Statistics}, vol.~42, no.~1, pp. 59--66, 1971.

\bibitem{neely2010stochastic}
M.~J. Neely, ``Stochastic network optimization with application to
  communication and queueing systems,'' \emph{Synthesis Lectures on
  Communication Networks}, vol.~3, no.~1, pp. 1--211, 2010.

\bibitem{hou2014scheduling}
I.-H. Hou, ``Scheduling heterogeneous real-time traffic over fading wireless
  channels,'' \emph{IEEE/ACM Transactions on Networking}, vol.~22, no.~5, pp.
  1631--1644, 2014.

\bibitem{kadota2018optimizing}
I.~Kadota, A.~Sinha, and E.~Modiano, ``Optimizing age of information in
  wireless networks with throughput constraints,'' in \emph{IEEE INFOCOM
  2018-IEEE Conference on Computer Communications}.\hskip 1em plus 0.5em minus
  0.4em\relax IEEE, 2018, pp. 1844--1852.

\bibitem{Jaramillo2011QOSinRealTimeTraffic}
J.~Jaramillo, R.~Srikant, and L.~Ying, ``\BIBforeignlanguage{English
  (US)}{Scheduling for optimal rate allocation in ad hoc networks with
  heterogeneous delay constraints},'' \emph{\BIBforeignlanguage{English
  (US)}{IEEE Journal on Selected Areas in Communications}}, vol.~29, no.~5, pp.
  979--987, 5 2011.

\bibitem{kang2013LDFperformanceOnRealTime}
\BIBentryALTinterwordspacing
X.~Kang, W.~Wang, J.~J. Jaramillo, and L.~Ying, ``On the performance of
  largest-deficit-first for scheduling real-time traffic in wireless
  networks,'' in \emph{Proceedings of the Fourteenth ACM International
  Symposium on Mobile Ad Hoc Networking and Computing}, ser. MobiHoc '13.\hskip
  1em plus 0.5em minus 0.4em\relax New York, NY, USA: ACM, 2013, pp. 99--108.
  [Online]. Available: \url{http://doi.acm.org/10.1145/2491288.2491298}
\BIBentrySTDinterwordspacing

\bibitem{Meko2015PacketLostOPT}
S.~F. Meko and H.~Seid, ``Stochastic approximation based scheduling for
  real-time applications in wireless networks,'' in \emph{AFRICON 2015}, Sep.
  2015, pp. 1--4.

\bibitem{G.Zhang2015TimelyThruputHCN}
G.~Zhang, A.~Huang, T.~Q.~S. Quek, and H.~Shan, ``Timely throughput of
  heterogeneous cellular networks,'' in \emph{2015 IEEE International
  Conference on Communications (ICC)}, June 2015, pp. 5621--5626.

\bibitem{Lashgari2013HeterNetworkTotalThruput}
S.~Lashgari and A.~S. Avestimehr, ``Timely throughput of heterogeneous wireless
  networks: Fundamental limits and algorithms,'' \emph{IEEE Transactions on
  Information Theory}, vol.~59, no.~12, pp. 8414--8433, Dec 2013.

\bibitem{singh2014fluctuation}
R.~Singh, I.-H. Hou, and P.~Kumar, ``Fluctuation analysis of debt based
  policies for wireless networks with hard delay constraints,'' in
  \emph{INFOCOM, 2014 Proceedings IEEE}.\hskip 1em plus 0.5em minus 0.4em\relax
  IEEE, 2014, pp. 2400--2408.

\bibitem{Kaul2011AoIinVehicular}
S.~Kaul, M.~Gruteser, V.~Rai, and J.~Kenney, ``Minimizing age of information in
  vehicular networks,'' in \emph{2011 8th Annual IEEE Communications Society
  Conference on Sensor, Mesh and Ad Hoc Communications and Networks}, June
  2011, pp. 350--358.

\bibitem{Modiano2017AoIinMHopsNetwork}
R.~Talak, S.~Karaman, and E.~Modiano, ``Minimizing age-of-information in
  multi-hop wireless networks,'' in \emph{2017 55th Annual Allerton Conference
  on Communication, Control, and Computing (Allerton)}, Oct 2017, pp. 486--493.

\bibitem{Modiano2018AoIBroadcast}
\BIBentryALTinterwordspacing
I.~Kadota, A.~Sinha, E.~Uysal{-}Biyikoglu, R.~Singh, and E.~Modiano,
  ``Scheduling policies for minimizing age of information in broadcast wireless
  networks,'' \emph{CoRR}, vol. abs/1801.01803, 2018. [Online]. Available:
  \url{http://arxiv.org/abs/1801.01803}
\BIBentrySTDinterwordspacing

\bibitem{Modiano2018PerfectChannelInfoAoIOptimal}
R.~Talak, S.~Karaman, and E.~Modiano, ``Optimizing age of information in
  wireless networks with perfect channel state information,'' in \emph{2018
  16th International Symposium on Modeling and Optimization in Mobile, Ad Hoc,
  and Wireless Networks (WiOpt)}, May 2018, pp. 1--8.

\bibitem{He2016AoILinkScheduling}
Q.~He, D.~Yuan, and A.~Ephremides, ``Optimizing freshness of information: On
  minimum age link scheduling in wireless systems,'' in \emph{2016 14th
  International Symposium on Modeling and Optimization in Mobile, Ad Hoc, and
  Wireless Networks (WiOpt)}, May 2016, pp. 1--8.

\end{thebibliography}


\begin{thebibliography}{00}
\bibitem{b1} G. Eason, B. Noble, and I. N. Sneddon, ``On certain integrals of Lipschitz-Hankel type involving products of Bessel functions,'' Phil. Trans. Roy. Soc. London, vol. A247, pp. 529--551, April 1955.
\bibitem{b2} J. Clerk Maxwell, A Treatise on Electricity and Magnetism, 3rd ed., vol. 2. Oxford: Clarendon, 1892, pp.68--73.
\bibitem{b3} I. S. Jacobs and C. P. Bean, ``Fine particles, thin films and exchange anisotropy,'' in Magnetism, vol. III, G. T. Rado and H. Suhl, Eds. New York: Academic, 1963, pp. 271--350.
\bibitem{b4} K. Elissa, ``Title of paper if known,'' unpublished.
\bibitem{b5} R. Nicole, ``Title of paper with only first word capitalized,'' J. Name Stand. Abbrev., in press.
\bibitem{b6} Y. Yorozu, M. Hirano, K. Oka, and Y. Tagawa, ``Electron spectroscopy studies on magneto-optical media and plastic substrate interface,'' IEEE Transl. J. Magn. Japan, vol. 2, pp. 740--741, August 1987 [Digests 9th Annual Conf. Magnetics Japan, p. 301, 1982].
\bibitem{b7} M. Young, The Technical Writer's Handbook. Mill Valley, CA: University Science, 1989.
\end{thebibliography}
\end{document}